\documentclass[submission,copyright,creativecommons]{eptcs}
\providecommand{\event}{ICE 2014} % Name of the event you are submitting to
\usepackage{breakurl}             % Not needed if you use pdflatex only.

\usepackage{graphicx}
\usepackage{color}

\usepackage{bussproofs} % use this instead of prooftree. see eg. at line 235

\usepackage{amssymb}
\usepackage{amsmath}

\usepackage{hyperref}

\newtheorem{theorem}{Theorem}[section]
\newtheorem{lemma}[theorem]{Lemma}

\newtheorem{definition}[theorem]{Definition}
\newtheorem{example}[theorem]{Example}
\newtheorem{remark}[theorem]{Remark}

\newcommand{\cp}[1]{}

\newenvironment{proof}[1][\!\!\,]{\vspace{1ex}\noindent\textbf{Proof #1: }}{\hfill$\Box$\vspace{2ex}}

\newcommand\psiTerms{\ensuremath{\mathbf{T}}}
\newcommand\psiConditions{\ensuremath{\mathbf{C}}}
\newcommand\psiAssertions{\ensuremath{\mathbf{A}}}
\newcommand\psiChanEq{\ensuremath{\stackrel{\ldotp}{\leftrightarrow}}}
\newcommand\psiCompAssert{\ensuremath{\otimes}}
\newcommand\psiAssertUnit{\ensuremath{\mathbf{1}}}
\newcommand\psiEntailment{\ensuremath{\vdash}}
\newcommand\psiassertion{\ensuremath{\Psi}}
\newcommand\psiEquivAssert{\ensuremath{\simeq}}
\newcommand\psiEmptyProc{\ensuremath{\mathbf{0}}}
\newcommand\psiOutput[2]{\ensuremath{\overline{#1}\langle#2\rangle}}
\newcommand\psiInput[2]{\ensuremath{\underline{#1}\langle#2\rangle}}
\newcommand\psiAssertOp[1]{\ensuremath{(\hspace{-0.07cm}\lvert#1\rvert\hspace{-0.07cm})}}
\newcommand\psiCase[1]{\ensuremath{\mathbf{case}\ #1\,}}
\newcommand\psiNew[1]{\ensuremath{(\nu #1)}}
\newcommand\psiframe[1]{\ensuremath{\mathcal{F}(#1)}}
\newcommand\psiPar{\ensuremath{\,|\,}}
\newcommand\psiRec{\ensuremath{!}}
\newcommand\psifresh{\ensuremath{\#}}
\newcommand\internalAction{\ensuremath{\tau}}

\newcommand\psitransition[4]{\ensuremath{#1\vartriangleright #2 \xrightarrow{#3} #4}}
\newcommand\eventpsi{\textsf{eventPsi}}
\newcommand{\dcrpsi}{\textsf{dcrPsi}}

\newcommand\defequal{\ensuremath{\stackrel{\vartriangle}{=}}}
\newcommand{\eqbydef}{\ensuremath{\stackrel{def}{=}}}

\newcommand{\condRelEv}{\ensuremath{\leq}}
\newcommand{\conflictRelEv}{\ensuremath{\sharp}}
\newcommand{\precondset}[1]{\ensuremath{\condRelEv\!\!#1}}
\newcommand{\preconfset}[1]{\ensuremath{\conflictRelEv #1}}
\newcommand{\composition}{\ensuremath{\otimes} }
\newcommand{\evtopsi}{\textsc{espsi}}
\newcommand{\concurrent}{\ensuremath{||}}
\newcommand{\transition}[1]{\ensuremath{\xrightarrow{#1}}}
\newcommand{\semanticsPsi}{\ensuremath{\vartriangleright}}
\newcommand{\partsof}[1]{\ensuremath{2^#1}}
\newcommand{\entailment}{\ensuremath{\vdash}}

\newcommand{\dcrtopsi}{\textsc{dcrpsi}}

\newcommand{\conditionRel}{\ensuremath{\rightarrow\!\!\!\bullet}}
\newcommand{\responseRel}{\ensuremath{\,\bullet\!\!\!\rightarrow}}
\newcommand{\includeRel}{\ensuremath{\rightarrow\!\!\!+}}
\newcommand{\excludeRel}{\ensuremath{\rightarrow\!\!\!\%}}
\newcommand{\milestoneRel}{\ensuremath{\rightarrow\!\!\!\diamond}}

\newcommand{\idealConfigs}[1]{\ensuremath{\lceil#1\rceil}}
\newcommand{\refinement}{\ensuremath{\mathit{ref}}}

\title{Concurrency Models with Causality and Events as Psi-calculi
}
\author{
\hspace{-4ex}H\aa{}kon Normann
\hspace{14ex}
Cristian Prisacariu
\thanks{This author was partially supported by the project \href{https://www.offpad.org/}{OffPAD} with number E!8324 part of the \href{http://www.eurekanetwork.org/activities/eurostars}{Eurostars} program funded by the \href{http://www.eurekanetwork.org}{EUREKA} and European Community.}
\institute{IT University of Copenhagen, \hspace{10ex} Dept. of Informatics, University of Oslo}
\email{\{haakno,cristi\}@ifi.uio.no}
\and 
Thomas Hildebrandt
\institute{IT University of Copenhagen, \ -- \ Rued Langgaardsvej 7, 2300 Copenhagen, Denmark}
\email{hilde@itu.dk}
}
\def\titlerunning{True Concurrency Models as Psi-calculi}
\def\authorrunning{Normann, Prisacariu, and Hildebrandt
}
\begin{document}
\maketitle

\begin{abstract}
Psi-calculi are a parametric framework for nominal calculi, where standard calculi are found as instances, like the pi-calculus, or the cryptographic spi-calculus and applied-pi. Psi-calculi have an interleaving operational semantics, with a strong foundation on the theory of nominal sets and process algebras. Much of the expressive power of psi-calculi comes from their logical part, i.e., assertions, conditions, and entailment, which are left quite open thus accommodating a wide range of logics. We are interested in how this expressiveness can deal with event-based models of concurrency. We thus take the popular prime event structures model and give an encoding into an instance of psi-calculi. We also take the recent and expressive model of Dynamic Condition Response Graphs (in which event structures are strictly included) and give an encoding into another corresponding instance of psi-calculi. The encodings that we achieve look rather natural and intuitive. Additional results about these encodings give us more confidence in their correctness.
\end{abstract}

\section{Introduction}\label{sec_intro}
% \vspace{-1ex}

Psi-calculi \cite{11psi_journal} are a recent framework where various existing calculi can be found as instances. In particular, the spi- and applied-pi calculi \cite{AbadiG99spi_calculus,AbadiF01applied_pi} are two instances of interest for security. Psi-calculi can also accommodate probabilistic models, by going through CC-pi \cite{BuscemiM07ccpi,NicolaFMPT05qosCCpi} which has already been treated as a corresponding psi-calculus instance.
The theory of psi-calculi is based on nominal data structures \cite{pitts_book_nominal}. 
Typed psi-calculus exists \cite{Huttel11typed_psi} as well as related instantiations as distributed pi-calculus \cite{huttel13private}.
Psi-calculi can be seen as a generalization of pi-calculus with two main features: \textit{(i)} nominal data structures (i.e., general, possibly open, terms) in place of communication channels and also in place of the communicated data; and \textit{(ii)} a rather open logic for capturing dependencies (i.e., through conditions and entailment) on the environment (i.e., assertions) of the processes.

The semantics of psi-calculi is given through structural operational rules and adopts an interleaving approach to concurrency, in the usual style of process algebras.
On the other hand, event-based models of concurrency take a non-interleaving view. These usually form domains and are used to give denotational semantics, as e.g., done by Winskel in \cite{Winskel82CCSevents,winskel95modelsCategory}.
Many times non-interleaving models of concurrency can actually distinguish between interleaving and, so called, ``true'' concurrency, as is the case with higher dimensional automata \cite{pratt91hda,Pratt00HDArev,Glabbeek06HDA}, configuration structures  \cite{GlabbeekP09configStruct}, or Chu spaces \cite{gupta94phd_chu,pratt95chu}.
The recent Dynamic Condition Response graphs (abbreviated DCR-graphs or DCRs) \cite{Hildebrandt10DCRs} is a model of concurrency with high expressive power which strictly extends event structures by refining the notions of dependent and conflicting events, and including the notion of response. Due to their graphical nature, DCRs have been successfully used in industry to model business processes \cite{SlaatsMHM13DCRsBusines}.

In this paper we are interested in how psi-calculi could accommodate the event structures model of concurrency \cite{NielsenPW79eventstructures,Winskel86}, with a final goal of capturing the DCRs model \cite{Hildebrandt10DCRs}.
Event names in event-based models of concurrency are unique, and can thus be thought of nominals, whereas the execution of an event can be seen as a communication or action of some sort. The dependencies between events that an event structure defines can be captured with rather simple assertions on nominal data structures, whereas the notion of computation is captured through reduction steps between psi-processes. To be confident on the encodings, we like to see a correlation between the notions of concurrency from the two encoded models and the interleaving diamonds from the psi-calculus behaviour.

These are the basic ideas we follow in this work to give encodings of event structures and DCRs into corresponding instances of psi-calculus. 
% The drawback is that psi-calculi have interleaving semantics using rewrite rules, whereas the event structures are a true concurrency model. 
After a couple of results meant to explain better the correlation between the encoding and the event structure model, we give a result that shows that the concurrency embodied by the event structure is captured in the encoding psi-process through the standard interleaving diamond.
For the event structures encoding we also give a result that identifies the syntactic shape of those psi-processes which correspond exactly to event structures.
Another feature of true concurrency models is that they are well behaved wrt.\ action refinement \cite{GlabbeekG01refinement}. For this we give a result showing that action refinement is preserved by our translation; under a properly defined refining operation on psi-processes, which we define similarly to the refinement operation on the event structures.
% A corollary of the refinement is that we obtain a composition result for a restricted form of parallel operation on event structures.

\vspace{-1ex}
\section{Background}\label{sec_background}

\vspace{-1ex}
\subsection{On psi-calculi}\label{subsec_background_psi}
% \vspace{-1ex}

\textit{Psi-calculus} \cite{11psi_journal} has been developed as a framework for defining nominal process calculi, like the many variants of the pi-calculus \cite{milner92picalcul}. The psi-calculi framework is based on nominal datatypes, \cite[Sec.2.1]{11psi_journal} giving an introduction to nominal sets used in psi-calculi. We will not explain much the nominal datatypes in this paper, but refer the reader to the book \cite{pitts_book_nominal} which contains a thorough treatment of both the theory behind nominal sets as well as various applications (e.g., see \cite[Ch.8]{pitts_book_nominal} for nominal algebraic datatypes). We expect, though, some familiarity with notions of algebraic datatypes and term algebras. 

The psi-calculi framework is parametric; instantiating the parameters accordingly, one obtains an \textit{instance of psi-calculi}, like the pi-calculus, or the cryptographic spi-calculus.
These parameters are:
\begin{center}
\begin{tabular}{ll}
\psiTerms & terms (data/channels)\\
\psiConditions & conditions\\
\psiAssertions & assertions
\end{tabular} 
\end{center}
which are nominal datatypes not necessarily disjoint; together with the following operators:
\begin{center}
\begin{tabular}{ll}
$\psiChanEq\ :\ \psiTerms\times\psiTerms\rightarrow\psiConditions$ & channel equality\\
$\psiCompAssert\ :\ \psiAssertions\times\psiAssertions\rightarrow\psiAssertions$ & composition of assertions\\
$\psiAssertUnit\ \in\ \psiAssertions$ & minimal assertion\\
$\psiEntailment\ \subseteq\ \psiAssertions\times\psiConditions$ & entailment relation
\end{tabular} 
\end{center}

Intuitively, terms can be seen as generated from a signature, as in term algebras; the conditions and assertions can be like in first-order logic; the minimal assertion being top/true, entailment the one from first-order logic, and composition taken as conjunction. 
It is helpful to think of assertions and conditions as logical formulas, and the entailment relation as an entailment in logic; but allow the intuition to think of logics abstractly, not just FOL, so that assertions and conditions are used to express any logical statements, where the entailment defines when assertions entail conditions (do not restrict to only thinking of truth tables; e.g., in our encodings we will use an extended logic for sets, with membership, pairs, etc.).
We will shortly exemplify how pi-calculus is instantiated in this framework.
The operators are usually written infix, i.e.: $M\psiChanEq N$, $\psiassertion \psiCompAssert\psiassertion'$, $\psiassertion\psiEntailment\varphi$.

The above operators need to obey some natural requirements, when instantiated. Channel equality must be symmetric and transitive. The composition of assertions must be associative, commutative, and have \psiAssertUnit\ as unit; moreover, composition must preserve equality of assertions, where two assertions are considered equal iff they entail the same conditions (i.e., for $\psiassertion,\psiassertion'\in\psiAssertions$ we define the equality  $\psiassertion\psiEquivAssert\psiassertion'$ iff $\forall\varphi\in\psiConditions:\psiassertion\psiEntailment\varphi \Leftrightarrow \psiassertion'\psiEntailment\varphi$).

The intuition is that assertions will be used to assert about the environment of the processes. Conditions will be used as guards for guarded (non-deterministic) choices, and are to be tested against the assertion of the environment for entailment. Terms are used to represent complex data communicated through channels, but will also be used to define the channels themselves, which can thus be more than just mere names, as in pi-calculus. The composition of assertions should capture the notion of combining assumptions from several components of the environment.

The syntax for building psi-process is the following (psi-processes are denoted by the $P,Q,\dots$; terms from \psiTerms\ by $M,N,\dots$):
\begin{center}
\begin{tabular}{ll}
$\psiEmptyProc$ & Empty/trivial process\\
$\psiOutput{M}{N}.P$ & Output\\
$\psiInput{M}{(\lambda\tilde{x})N}.P$ & Input\\
$\psiCase{\varphi_{1}:P_{1},\dots,\varphi_{n}:P_{n}}$ & Conditional (non-deterministic) choice\\
$\psiNew{a}P$ & Restriction of name $a$ inside processes $P$\\
$P\psiPar Q$ & Parallel composition\\
$\psiRec P$ & Replication\\
$\psiAssertOp{\psiassertion}$ & Assertions
\end{tabular} 
\end{center}

The input and output processes are as in pi-calculus only that the channel objects $M$ can be arbitrary terms. In the input process the object $(\lambda\tilde{x})N$ is a pattern with the variables $\tilde{x}$ bound in $N$ as well as in the continuation process $P$. Intuitively, any term message received on $M$ must match the pattern $N$ for some substitution of the variables $\tilde{x}$. The same substitution is used to substitute these variables in $P$ after a successful match.
The traditional pi-calculus input $a(x).P$ would be modelled in psi-calculi as $\psiInput{a}{(\lambda x)x}.P$, where the simple names $a$ are the only terms allowed.
Restriction, parallel, and replication are the standard constructs of pi-calculus.

The case process behaves like one of the $P_{i}$ for which the condition $\varphi_{i}$ is entailed by the current environment assumption, as defined by the notion of \textit{frame} which we present later. This notion of frame is familiar from the applied pi-calculus, where it was introduced with the purpose of capturing static information about the environment (or seen in reverse, the frame is the static information that the current process exposes to the environment).
A particular use of case is as $\psiCase{\varphi:P}$ which can be read as $\mathbf{if}\ \varphi\ \mathbf{then}\ P$. Another special usage of case is as $\psiCase{\top:P_{1},\,\top:P_{2}}$, where $\psiassertion\psiEntailment\top$ is a special condition that is entailed by any assertion, like $a\psiChanEq a$; this use is mimicking the pi-calculus non-deterministic choice $P_{1} + P_{2}$.
Infinite summation is sometimes found in process algebras, e.g., in Milner's SCCS \cite{Milner83SCCS}. In the case of psi-calculi an infinite case construct can be used as $\psiCase{\tilde{\varphi_{i}}:\tilde{P_{i}}}$ where we use infinite lists to represent the respective condition/process pairs. There is no change to the semantics. The same semantics works for infinite parallel processes as well; though the replication is the preferred way to obtain infinite parallel components.

Assertions $\psiAssertOp{\psiassertion}$ can float freely in a process (i.e., be put in parallel) describing assumptions about the environment. Otherwise, assertions can appear at the end of a sequence of input/output actions, i.e., these are the guarantees that a process provides after it makes an action (on the same lines as in assume/guarantee reasoning about programs). Assertion processes are somehow similar to the active substitutions of the applied pi-calculus, only that assertions do not have computational behaviour, but only restrict the behaviour of the other constructs by providing their assumptions about the environment.

\begin{example}[pi-calculus as an instance]\label{ex_pi_instance}
To obtain pi-calculus \cite{milner92picalcul} as an instance of psi-calculus use the following, built over a single set of names $\mathcal{N}$:
\begin{center}
\begin{tabular}{lcl}
\psiTerms & \defequal & $\mathcal{N}$\\
\psiConditions & \defequal & $\{a=b \mid a,b\in\psiTerms\}$\\
\psiAssertions & \defequal & $\{\psiAssertUnit\}$\\
$\psiChanEq$ &  \defequal & $=$\\
% $\psiCompAssert\ :\ \psiAssertions\times\psiAssertions\rightarrow\psiAssertions$ & composition of assertions\\
% $\psiAssertUnit$ & \defequal & $1$\\
$\psiEntailment$ &  \defequal & $\{(\psiAssertUnit,a=a) \mid a\in \psiTerms\}$
\end{tabular} 
\end{center}
with the trivial definition for the composition operation.
The only terms are the channel names $a\in\mathcal{N}$, and there is no other assertion than the unit. The conditions are equality tests for channel names, where the only successful tests are those where the names are equal. Hence, channel comparison is defined as just name equality.
\end{example}

Psi-calculus is given an operational semantics in \cite{11psi_journal} using labelled transition systems, where the nodes are the process terms and the transitions represent one reduction step, labelled with the action that the process executes. The actions, generally denoted by $\alpha,\beta$, represent respectively the input and output constructions, as well as $\internalAction$ the internal synchronization/communication action: 
% \[
% \psiOutput{M}{(\nu\tilde{a})N} \mid \psiInput{M}{N} \mid \internalAction
% \]
% 

\vspace{1ex}\centerline{$\psiOutput{M}{(\nu\tilde{a})N} \mid \psiInput{M}{N} \mid \internalAction$}

% \newpage

Transitions are done in a context, which is represented as an assertion \psiassertion, capturing assumptions about the environment:

\centerline{$\psitransition{\psiassertion}{P}{\alpha}{P'}$}

\noindent Intuitively, the above transition could be read as: The process $P$ can perform an action $\alpha$ in an environment respecting the assumptions in $\psiassertion$, after which it would behave like the process $P'$.

The context assertion is obtained using the notion of \textit{frame} which essentially collects (using the composition operation) the outer-most assertions of a process.
The frame $\psiframe{P}$ is defined inductively on the structure of the process as:
\begin{center}
\begin{tabular}{l}
$\psiframe{\psiAssertOp{\psiassertion}} = \psiassertion$\\
$\psiframe{P\psiPar Q} = \psiframe{P} \psiCompAssert \psiframe{Q}$\\
$\psiframe{\psiNew{a}{P}} = \psiNew{a}{\psiframe{P}}$\\
$\psiframe{\psiRec P} = \psiframe{\psiCase{\tilde{\varphi}:\tilde{P}}} = \psiframe{\psiOutput{M}{N}.P} = \psiframe{\psiInput{M}{(\lambda\tilde{x})N}.P} = \psiAssertUnit$
\end{tabular} 
\end{center}
Any assertion that occurs under an action prefix or a condition is not visible in the frame.

We give only an exemplification of the transition rules for psi-calculus, and refer to \cite[Table 1]{11psi_journal} for the full definition. The \textsc{(case)} rule shows how the conditions are tested against the context assertions. The communication rule \textsc{(com)} shows how the environment processes executing in parallel contribute their top-most assertions to make the new context assertion for the input-output action of the other parallel processes.
In the \textsc{(com)} rule the assertions $\psiassertion_{P}$ and $\psiassertion_{Q}$ come from the frames of $\psiframe{P}=\psiNew{\tilde{b}_{P}}{\psiassertion_{P}}$ respectively $\psiframe{Q}=\psiNew{\tilde{b}_{Q}}{\psiassertion_{Q}}$.
In \textsc{(par)} $bn(\alpha)\psifresh Q$ says that the bound names of $\alpha$ are fresh in $Q$.

\begin{center}
\vspace{2ex}
\AxiomC{$\psiassertion \psiEntailment M\psiChanEq K$}
\RightLabel{\textsc{(inn)}}
\UnaryInfC{${\psitransition{\psiassertion}{\psiInput{M}{(\lambda \tilde{y}) N}.P}{\underline{K}N[\tilde{y}:= \tilde{L}]}{P[\tilde{y}:= \tilde{L}]}}$}
\DisplayProof
\hspace{3ex}
\vspace{2ex}
\AxiomC{$\psiassertion \psiEntailment M\psiChanEq K$}
\RightLabel{\textsc{(out)}}
\UnaryInfC{$\psitransition{\psiassertion}{\psiOutput{M}{N}.P}{\overline{K}N}{P}$}
\DisplayProof
\vspace{2ex}
% 
% \begin{prooftree}
\AxiomC{$\psitransition{\psiassertion}{P_{i}}{\alpha}{P'}$}
\AxiomC{$\psiassertion\psiEntailment\varphi_{i}$}
% \AxiomC{$\psiassertion_{Q}\psiCompAssert\psiassertion_{P}\psiCompAssert\psiassertion\psiEntailment M\psiChanEq K$}
\RightLabel{\textsc{(case)}}
% \TrinaryInfC{$\psitransition{\psiassertion}{P\psiPar Q}{\internalAction}{\psiNew{\tilde{a}}{(P'\psiPar Q')}}$}
\BinaryInfC{$\psitransition{\psiassertion}{\psiCase{\tilde{\varphi}:\tilde{P}}}{\alpha}{P'}$}
\DisplayProof
\AxiomC{$\psitransition{\psiassertion\psiCompAssert\psiassertion_Q}{P}{\alpha}{P'}$}
\AxiomC{$bn(\alpha)\psifresh Q$}
\RightLabel{\textsc{(par)}}
\BinaryInfC{$\psitransition{\psiassertion}{P\psiPar Q}{\alpha}{P'\psiPar Q}$}
\DisplayProof
% \end{prooftree}
% 
% 
\AxiomC{$\psitransition{\psiassertion}{ P|!P}{\alpha}{P'}$}
\RightLabel{\textsc{(rep)}}
\UnaryInfC{$\psitransition{\psiassertion}{ !P}{\alpha}{P'}$}
\DisplayProof
\AxiomC{$\psitransition{\psiassertion_{Q}\psiCompAssert\psiassertion}{P}{\overline{M}(\nu\tilde{a})N}{P'} $}
\AxiomC{$\psitransition{\psiassertion_{P}\psiCompAssert\psiassertion}{Q}{\underline{K}N}{Q'} $}
\AxiomC{$\psiassertion_{Q}\psiCompAssert\psiassertion_{P}\psiCompAssert\psiassertion\psiEntailment M\psiChanEq K$}
\RightLabel{\textsc{(com)}}
\TrinaryInfC{$\psitransition{\psiassertion}{P\psiPar Q}{\internalAction}{\psiNew{\tilde{a}}{(P'\psiPar Q')}}$}
\DisplayProof
\end{center}
% \[
% % \infrule{(case)}{\psitransition{\psiassertion}{P_{i}}{\alpha}{P'} \hspace{-1ex}\andalso \psiassertion\psiEntailment\varphi_{i}}{\psitransition{\psiassertion}{\psiCase{\tilde{\varphi}:\tilde{P}}}{\alpha}{P'}}
% \hspace{-0.5ex}\infrule{(par)}{\psitransition{\psiassertion\psiCompAssert\psiassertion_Q}{P}{\alpha}{P'} \hspace{-1ex}\andalso bn(\alpha)\#Q}
% {\psitransition{\psiassertion}{P\psiPar Q}{\alpha}{P'\psiPar Q}}
% \hspace{-0.5ex}\infrule{(out)}{\psiassertion \psiEntailment M\psiChanEq K}
% {\psitransition{\psiassertion}{\overline{M}N.P}{\overline{K}N}{P}}
% \]
% 
% \[
% \infrule{(inn)}{\psiassertion \psiEntailment M\psiChanEq K}
% {{\psitransition{\psiassertion}{\underline{M}(\lambda \tilde{y}) N.P}{\underline{K}N[\tilde{y}:= \tilde{L}]}{P[\tilde{y}:= \tilde{L}]}}}
% \infrule{(rep)}{\psitransition{\psiassertion}{ P|!P}{\alpha}{P'}}{\psitransition{\psiassertion}{ !P}{\alpha}{P'}}
% \]
% \[
% \infrule{(com)}{
% \psitransition{\psiassertion_{Q}\psiCompAssert\psiassertion}{P}{\overline{M}(\nu\tilde{a})N}{P'} 
% \andalso 
% \psitransition{\psiassertion_{P}\psiCompAssert\psiassertion}{Q}{\underline{K}N}{Q'} 
% \andalso
% \psiassertion_{Q}\psiCompAssert\psiassertion_{P}\psiCompAssert\psiassertion\psiEntailment M\psiChanEq K
% }{
% \psitransition{\psiassertion}{P\psiPar Q}{\internalAction}{\psiNew{\tilde{a}}{(P'\psiPar Q')}}
% }
% \]

There is no transition rule for the assertion process; this is only used in constructing frames. Once an assertion process is reached, the computation stops, and this assertion remains floating among the other parallel processes and will be composed part of the frames, when necessary, like in the case of the communication rule.
The empty process has the same behaviour as, and thus can be modelled by, the trivial assertion $\psiAssertOp{\psiAssertUnit}$.

\subsection{On event structures}\label{subsec_ES_DCRs}

For event structures we try to follow the standard notation and terminology from \cite[sec.8]{winskel95modelsCategory}.

\begin{definition}[prime event structures]\label{def_eventStructures}
% A \emph{prime event structure} is a tuple $\mathcal{E} = (E,\condRelEv,\conflictRelEv)$ where $E$ is a set of events, $\condRelEv\hspace{3pt}\subseteq E\times E$ is a partial order (the \emph{causality} relation) satisfying the \emph{principle of finite causes}, i.e., $\forall e\in E:\{d\in E|d\condRelEv e\}$ is finite, and $\sharp \subseteq E\times E$ is an irreflexive, symmetric binary relation (the \emph{conflict} relation) satisfying the \emph{principle of conflict heredity}, i.e., $\forall d, e, f \in E: d\condRelEv e\wedge d\conflictRelEv f \Rightarrow e\conflictRelEv f$.
\emph{A labelled prime event structure} over alphabet Act is a tuple $\mathcal{E} = (E,\condRelEv,\conflictRelEv,l )$ where $E$ is a possibly infinite set of events, $\condRelEv\hspace{3pt}\subseteq E\times E$ is a partial order (the \emph{causality} relation) satisfying
\begin{enumerate}
\item\label{es_finiteCauses} 
 \emph{the principle of finite causes}, i.e.: $\forall e\in E:\{d\in E \mid d\condRelEv e\}$ is finite,
\end{enumerate}
and $\sharp \subseteq E\times E$ is an irreflexive, symmetric binary relation (the \emph{conflict} relation) satisfying
\begin{enumerate}
\setcounter{enumi}{1}
\item\label{es_conflictHer} 
\emph{the principle of conflict heredity}, i.e., $\forall d, e, f \in E: d\condRelEv e\wedge d\conflictRelEv f \Rightarrow e\conflictRelEv f$.
\end{enumerate}
and $l:E\rightarrow Act$ is the labelling function. 
Denote by \textbf{E} the set of all prime event structures. 
\end{definition}

Intuitively, a prime event structure models a concurrent system by taking $d\condRelEv e$ to mean that event $d$ is a prerequisite of event $e$, i.e., event $e$ cannot happen before event $d$ has been done. A conflict $d\conflictRelEv e$ says that events $d$ and $e$ cannot both happen in the same run. 

\begin{definition}[concurrency]\label{def_conc_ES}
\emph{Casual independence (concurrency)} between events is defined in terms of the above two relations as 
\[
d\concurrent e \defequal \neg(d\condRelEv e\vee e\condRelEv d\vee d\conflictRelEv e)
\]
capturing the intuition that two events are concurrent when there is no causal dependence between the two and they are not in conflict.
\end{definition}

The behaviour of an event structure is described by subsets of events that happened in some (partial) run. This is called a \textit{configuration} of the event structure, and \textit{steps} can be defined between configurations.
% which can also be viewed as a state of the event structure. 

\begin{definition}[configurations]\label{def_configs_ES}
Define a \emph{configuration} of an event structure $\mathcal{E} = (E,\condRelEv,\conflictRelEv)$ to be a finite subset of events $C\subseteq E$ that respects: 
\begin{enumerate}
\item \emph{conflict-freeness}: $\forall e,e'\in C: \neg(e\conflictRelEv e')$ and,
\item \emph{downwards-closure}: $\forall e,e'\in E: e'\condRelEv e \wedge e\in C \Rightarrow e'\in C$.
\end{enumerate}
We denote the set of all configurations of some event structure by $\mathcal{C}_{\mathcal{E}}$.
\end{definition}

Note in particular that $\emptyset$ is a configuration (i.e., the root configuration) and that any set $\idealConfigs{e}\defequal\{e'\in E \mid e'\condRelEv e\}$ is also a configuration determined by the single event $e$.
Events determine steps between configurations in the sense that $C \transition{e} C'$ whenever $C,C'$ are configurations, $e\not\in C$, and $C'=C\cup \{e\}$.

\begin{remark}\label{remark_ES}
It is known (see e.g., \cite[Prop.18]{winskel95modelsCategory}) that prime event structures are fully determined by their sets of configurations, i.e., the relations of causality, conflict, and concurrency can be recovered only from the set of configurations $\mathcal{C}_{\mathcal{E}}$ as follows:
\begin{enumerate}
\item $e\condRelEv e'$ iff $\forall C\in\mathcal{C}_{\mathcal{E}}: e'\in C \Rightarrow e\in C$;
\item $e\conflictRelEv e'$ iff $\forall C\in \mathcal{C}_{\mathcal{E}}: \neg(e\in C \wedge e'\in C)$;
\item $e\concurrent e'$ iff $\exists C,C'\in\mathcal{C}_{\mathcal{E}}: e\in C \wedge e'\not\in C \wedge e'\in C' \wedge e\not\in C' \wedge C\cup C'\in \mathcal{C}_{\mathcal{E}}$.
\end{enumerate}
\end{remark}

% \begin{notation}
For some event $e$ we denote by $\precondset{e} = \{e'\in E \mid e' \condRelEv e\}$ the set of all events which are conditions of $e$ (which is the same as the notation $\idealConfigs{e}$ from \cite{winskel95modelsCategory}, but we prefer to use the above so to be more in sync with similar notations we use in this paper for similar sets defined for DCRs too), and $\preconfset{e}=\{e'\in E \mid e' \conflictRelEv e \}$ those events in conflict with $e$.
% \end{notation}

\subsection{On DCR-graphs}\label{subsec_dcr}

Dynamic Condition Response graphs (DCR-graphs) is a recent model of concurrency, which generalizes event structures by taking into account progress in terms of demanded responses, while giving a finite model of possibly infinite behaviour. Using a graphic notation along with the formal, it is already used in industry for workflow management. 
We follow the notations for DCRs from 
% \hn{would like maybe use <Modular Context-Sensitive and Aspect-Oriented Processes with Dynamic Condition Response Graphs> instead as newer article and they made some notation changes since the nested one.} 
\cite{Hildebrandt10DCRs,HildebrandtMS11DCRs}. 
 
% % old
%  \begin{definition}[DCR Graphs]
%  A Dynamic Condition Response Graph (DCR Graph) G is a tuple $(E,M,\conditionRel,\responseRel,\milestoneRel,\includeRel,\excludeRel,L,l)$ where
%  \begin{enumerate}
%  \item E is a set of events
%  \item M is the marking, $\mathcal{P}(E)\times\mathcal{P}(E)\times\mathcal{P}(E)$,
%  \item $\conditionRel,\responseRel,\milestoneRel,\includeRel,\excludeRel\subseteq E\times E$ is the conditions, response, milestone, include and exclude relations respectively. 
%  \item L is the set of labels and $l:E\rightarrow\mathcal{P}(L)$ is a labelling function mapping events to labels.
%  \end{enumerate}
%  \end{definition}
 \begin{definition}[DCR Graphs]
 We define a \emph{Dynamic Condition Response Graph} to be a tuple $G=(E,M,\conditionRel,\responseRel,\milestoneRel,\includeRel,\excludeRel,L,l)$ where
 \begin{enumerate}
 \item $E$ is a set of events,
 \item $M\in \partsof{E}\times \partsof{E}\times \partsof{E}$ is the initial marking,
 \item $\conditionRel,\responseRel,\milestoneRel,\includeRel,\excludeRel\subseteq E\times E$ are respectively called the condition, response, milestone, include, and exclude relations, 
 \item $l:E\rightarrow L$ is a labelling function mapping events to labels from $L$.
 \end{enumerate}
 \end{definition}

% More recent works on DCR graphs \cite{HildebrandtMS11DCRs} enrich the model with a \textit{milestone} relation $\milestoneRel$. For the purpose of having a more simple exposition we prefer to present our results on DCRs as defined above. Treating the milestone relation in the encoding of Section~\ref{sec_DCR_into_psi} does not pose difficulties, and we only point out the necessary changes.
 
For any relation $\rightarrow\in\{\conditionRel,\responseRel,\milestoneRel,\includeRel,\excludeRel\}$, we use the notation $e\rightarrow$ for the set $\{e'\in E \mid e\rightarrow e'\}$ and $\rightarrow e$ for the set $\{e'\in E \mid e'\rightarrow e\}$ of events $e'\in E$ which are in the respective relation with $e$. 
 
A marking $M=(Ex,Re,In)$ represents a state of the DCR. One should understand $Ex$ as the set of \textit{executed} events, $Re$ the set of \textit{response} events that must happen sometime in the future, and $In$ the set of \textit{included} events, i.e., those that \textit{may} happen in the next steps. The five relations impose constraints on the events and dictate the dynamic inclusion and exclusion of events. 
 
For a DCR graph $(E, M, \conditionRel,\responseRel,\milestoneRel,\includeRel,\excludeRel)$ and a marking $M=(Ex,Re,In)$, we say that an \textit{event $e\in E$ is enabled in $M$}, written $M\vdash e$, iff $e\in In\wedge(In\cap \conditionRel e)\subseteq Ex\wedge(In\cap\milestoneRel e)\subseteq E\setminus Re$. 
Intuitively, an event can only happen if it is included, all its included preconditions have been executed, and 
none of the included events that are milestones for it are scheduled responses.
% is not waiting on some milestones to finish.
The behaviour of a DCR is given through transitions between markings done by executing enabled events. The result of the execution of the event $e$ in marking $M=(Ex,Re,In)$ is defined as the new marking $M' \eqbydef (Ex\cup\{e\}, (Re\setminus\{e\})\cup e\responseRel,\newline (In\setminus e\excludeRel)\cup e\includeRel)$. We denote a transition as $M\transition{e}M'$.
 %
 % \hn{To use the older transition system version of execution rule one would have to change up the labels and as it would then become a mix between the old papers with roles etc, and the new where those has been removed. Non of the newer ones have a definition of non milestones nor the labelled transition system. Only the definition of when an event is enabled and what happens to marking when it does like this definition does (the newest i have also have this exact definition and i would say to use newest and not mix definitions between papers would probably be the best.)} 
% 
An event can happen an arbitrary number of times as long as it is enabled. Events that should happen only once must explicitly be excluded. %From a workflow perspective one can say that same action may happen several times or even be demanded several times, if for example someone done something wrong and need to redo to fix it. 
 
An event structure $(E,\condRelEv,\conflictRelEv,l)$ is a special case of a DCR graph $(E,M,\condRelEv,\emptyset,\emptyset,\emptyset,\conflictRelEv\cup id)$ where each event is excluding itself, i.e., cannot be done multiple times, and the conflict relation is modelled by mutual exclusion. The response, include, and milestone relations are empty, and initially all events are included, as the marking $M=(\emptyset,\emptyset,E)$, i.e., all events can be executed; this comes from \cite[Prop.1\&3]{Hildebrandt10DCRs}.
Essentially, the conflict relation excludes all related events; and the causality relation is the condition relation of the DCR. The rest of the DCR relations are just additions wrt.\ the event structures model, therefore should be empty. Moreover, the initial marking has no executed events and no responses, but all events are initially included.
Opposed to the behaviour of event structures, in full DCRs we also have that the causality between events can change during the run, as events are included or excluded. Moreover, the conflict in DCRs is not permanent as is the case with event structures or with the various proposals of cancellation of Pratt. Conflict in DCR can be transient since an event can be included and excluded during a run.
So, already at the conflict and causality relations, the DCRs depart from event structures in a non-trivial manner.

DCRs have peculiar aspects which offer them good expressive power that proved useful in various practical situations, like for business workflows. But we are not concerned with explaining or motivating these more, as the related literature does a much better job. We are concerned with finding a nice and intuitive encoding of DCRs in the expressive psi-calculi framework.

\section{Encoding event structures in psi-calculi}\label{sec_evnet_into_psi}

% We have taken in Definition~\ref{def_eventStructures}, the quite popular version of prime event structures with nice properties like correlations with domains which makes them a good candidate for being used for semantics of concurrent programs. Due to this popularity we chose to encode this version of event structures in this section. Nevertheless, we believe that other, more general, versions of event structures, like those from \cite{Winskel86} or \cite{GlabbeekP09configStruct}, can be encoded in psi-calculi following similar ideas as we give here.
Due to their popularity, we have chosen to encode, in this section, the version of event structures called \textit{prime} as defined in Definition~\ref{def_eventStructures}. These have many nice features like correlations with domains which makes them a good candidate for being used for denotational semantics of concurrent programs. Nevertheless, we believe that other, more general, versions of event structures, like those from \cite{Winskel86} or \cite{GlabbeekP09configStruct}, can be encoded in psi-calculi following similar ideas as we give here.

\begin{definition}[event psi-calculus]\label{def_event_psi}
We define a psi-calculus instance, called \emph{\eventpsi}, parametrized by a nominal set $E$, to be understood as \emph{events}, by providing the following definitions of the key elements of a psi-calculus instance:
\vspace{-1ex}\[\psiTerms\eqbydef E
%\hn{\psiTerms\eqbydef E\leftarrow \mathbb{N} < }
\hspace{5ex} \psiConditions\eqbydef \partsof{E} \times \partsof{E}
\hspace{5ex} \psiAssertions\eqbydef \partsof{E}
\hspace{5ex} \psiChanEq\eqbydef=
\hspace{5ex} \psiCompAssert\eqbydef\cup
\hspace{5ex} \psiAssertUnit\eqbydef\emptyset
\vspace{-1ex}\]
\[\psiEntailment\eqbydef \Psi\psiEntailment\varphi\hspace{3pt}\textnormal{iff}\hspace{3pt}(\pi_{L}(\varphi)\subseteq\Psi)\wedge(\Psi\cap\pi_{R}(\varphi)=\emptyset) 
\hspace{5ex} \Psi\psiEntailment a\psiChanEq b \hspace{3pt}\textnormal{iff}\hspace{3pt} a = b
\]
where $\psiTerms$, $\psiConditions$, and $\psiAssertions$ are nominal data types built over the nominal set $E$, and $\pi_{L},\pi_{R}$ are the standard left/right projection functions for pairs. 
Denote by $ \mathit{en}(P)\!\!\subseteq\!E$ the event names appearing in a process $P$.
\end{definition}

The conditions \psiConditions\ are pairs of subsets of events, which intuitively will hold the enabling conditions for an event, i.e., the left set holding those events it depends on and the right set holding those events it is in conflict with. The assertions \psiAssertions\ intuitively can be understood as capturing the set of all executed events, i.e., a configuration of the event structure. Channel equivalence is equality of event names, as in standard pi-calculus. 
Composition of two assertions is the union of the sets. 
% Identity is the empty-set. 
The entailment \entailment\ intuitively captures when events may fire, thus describing when events are enabled by a configuration.

It is easy to see that our definitions respect the restrictions of making a psi-calculus instance. In particular, channel equivalence is symmetric and transitive since equality is. The \composition\ is compositional, associative and commutative, as $\cup$ is; and moreover $\emptyset\cup S = S$, for any set S, i.e., $\mathbf{1}$ is the identity.

\begin{definition}[event structures to \eventpsi]\label{def_evtopsi}
We define a function $\evtopsi$ which given an event structure $\mathcal{E} = (E, \condRelEv,\conflictRelEv)$ and a configuration $C$ of $\mathcal{E}$, returns an $\eventpsi$-process $P_{E}=|_{e\in E} P_{e}$ with $P_e =\psiAssertOp{\{e\}}$ if $e\in C$, otherwise $P_e = \psiCase{\varphi_e: \psiOutput{e}{e}.\psiAssertOp{\{e\}}}$, where $\varphi_e = (\precondset{e}, \preconfset{e})$.
\end{definition}

A process generated by the \evtopsi\ function is built up from smaller ``event processes'' put in parallel. These come in two forms: those corresponding to the events in the configuration of the translated event structure (i.e., those that  already happened), and processes corresponding to events that have not happened yet. 
For the latter we use a condition $\varphi_e$ that contains the set $\precondset{e}$ of events $e$ is depending on and the set $\preconfset{e}$ of events $e$ is in conflict with. Together these two sets along with the frame of the entire psi-process, decide, through the entailment, if the event can execute or not. 
When an event happens we will have a transition over the channel with the same name as the event.
Usually an event structure is encoded into \eventpsi\ starting from the empty configuration, i.e., with no behaviour.

The set \psiTerms\ may be infinite, hence elements of \psiAssertions\ and \psiConditions\ may be infinite terms (sets). In the encoding produced by \evtopsi, the conditions have $\pi_{L}(\varphi)$ finite, because of the principle of finite causes of Definition~\ref{def_eventStructures}.1 that event structures respect.
Still, the $\pi_{R}(\varphi)$ may be infinite, because there is no restriction on the conflict relation in event structures, and thus an event can be in conflict with infinitely many events, therefore \evtopsi\ may create infinite condition terms.

An intuitive example where this would appear is when we model looping behaviour of a system with event structures, and we have a looping branch, which would be unfolded into infinitely many sequential events, and we have a second branch which cancels this looping branch (i.e., as with a choice). The cancelling of the looping branch would mean cancelling all the infinitely many events that encode this branch. That is to say, the single event is in conflict with all the events on the looping branch.

Assertion terms from \psiAssertions, produced by \evtopsi, are always finite because they encode, cf.\ Lemma~\ref{conf_maint}, configurations, which are finite sets. Therefore, it is not problematic to have the infinite part of the conditions, since the only place where this is used is in deciding the entailment, which would thus always terminate, hence be decidable for any assertion/configuration used in the encoding.

Besides this, the encoding \evtopsi\ builds in parallel infinitely many processes, one for each $e\in E$.
For practical reasons infinite terms are not desired. 
% But there are works with infinite terms, like infinite summation. Recursion of pi-calculi and replication of psi-calculus are the way to obtain infinite parallel components, as we generate from ES.
% 
% In logics it is sometimes the case that one works with infinite formulas, e.g., considering infinite conjunctions (opposed to standard finite conjunctions). Such infinite formulas usually make the presentation more nice.
% 
But there are works with infinite terms, like infinite summation in SCCS, infinite case construct for psi-calculus, or infinite conjunctions in some logics. Such infinite formulas usually make the presentation more nice.
In our case we also wanted to have the nice presentation, therefore we opted to generate infinite terms. From our terms it is clear to see the correlation with the event structures. We work the same as in event structures, by tacitly having infinite events, thus infinite parallel processes.
Encoding the infinite terms with the replication (i.e., one replication of an infinite case construct) would make the presentation more cluttered, with the details easily becoming unpleasant.

% % % 

We could say that prime event structures are ``wildly'' infinite. If we would otherwise take a kind of event structures that are regular, i.e., are build from some operations like choice and sequence, and the infinity comes only from some recursion operation, then we think that this infinity could be encoded with the finite apparatus of psi-calculi.
But it is not clear which event structures are ``regular''; and for our purposes the prime event structures are a good enough concurrency model to look at.

Our intention is to investigate the expressive power of the psi-calculi framework; the power of its logical part, i.e., the assertions, conditions, and entailment, and the complex nominal data structures that can be used both for communication and for transmitted data. 
% In particular to correlate the expressiveness of true concurrency models with psi-calculi.
% % % 

\begin{lemma}[correspondence configuration--frame]\label{conf_maint}
For any event structure $\mathcal{E}$ and configuration $C_{\mathcal{E}}$, the frame of the \eventpsi-process $\evtopsi(\mathcal{E},C_{\mathcal{E}})$ corresponds to the configuration $C_{\mathcal{E}}$. 
\end{lemma}

% \vspace{-1ex}
\begin{proof}
Denote $\evtopsi(\mathcal{E},C_{\mathcal{E}})=P_{E}$ as in Definition~\ref{def_evtopsi}.
The frame of $P_{E}$ is the composition with \composition\ of the frames of $P_{e}$ for $e\in E$. 
%\cp{Note that event structures may have infinitely many events, so we talk as in the definition about the events in E and not indexed by some bounded i 1-n.} 
As $P_e$ is either $\psiAssertOp{\{e\}}$ if $e\in C_\mathcal{E}$ or $\psiCase\varphi_e: \psiOutput{e}{e}.\psiAssertOp{\{e\}}$ then the frame of $P_e$ would be either $\psiframe{\psiAssertOp{\{e\}}} = \{e\}$ or $\psiframe{\psiCase\varphi_e: \psiOutput{e}{e}.\psiAssertOp{\{e\}}} = \psiAssertUnit$. Thus the frame of $P_E$ is the \composition  of \psiAssertUnit's and all events in $C_{\mathcal{E}}$, thus having that the frame is the union of all events in $C_{\mathcal{E}}$
\end{proof}

\vspace{-1ex}
\begin{lemma}[transitions preserve configurations]\label{lemma_conf_maint2}
For some event structure $\mathcal{E}$ and some configuration of it $C_\mathcal{E}$, any transition from this configuration $C_\mathcal{E}\transition{e}C_\mathcal{E}'$ is matched by a transition $\emptyset\,\triangleright \evtopsi(\mathcal{E},C_\mathcal{E})\transition{\overline{e}e}\evtopsi(\mathcal{E},C_\mathcal{E}')$ in the corresponding \eventpsi-process.
The other way, any transition $\emptyset\,\triangleright \evtopsi(\mathcal{E},C_\mathcal{E})\transition{\overline{e}e}P'$ is matched by a step $C_\mathcal{E}\transition{e}C_\mathcal{E}'$, with $P'=\evtopsi(\mathcal{E},C_\mathcal{E}')$.
\end{lemma} 

\vspace{-1ex}
\begin{proof}
Before the event $e$ is executed we have that our \eventpsi-process $\evtopsi(\mathcal{E},C_\mathcal{E})$ can we written in the form $P = \psiCase{\varphi_e:\psiOutput{e}{e}.\psiAssertOp{\{e\}}}|Q$.
By Lemma~\ref{conf_maint} we know that the frame of $P$ is the same as $C_\mathcal{E}$, i.e., we have that $\psiframe{P} = \psiAssertUnit\psiCompAssert\psiframe{Q} = \Psi_Q = C_\mathcal{E}$ before $e$ has happened, and $e\notin C_\mathcal{E}$.

We can observe the transition between \eventpsi-processes by the following proof tree, using the transition rules of psi-calculi.
% 
% \[
% \infrule{Par}{
%    \infrule{Case}{
%      \infrule{out}{
%        \Psi_Q\composition\emptyset\,\vdash e\stackrel{\cdot}{\leftrightarrow} e
%      }{
%        \Psi_Q\composition\emptyset\,\triangleright\psiOutput{e}{e}.\psiAssertOp{\{e\}}\stackrel{\overline{e}e}{\rightarrow}\psiAssertOp{\{e\}}
%      }
%      \Psi_Q\composition\emptyset\,\vdash\varphi_e
%    }{
%      \Psi_Q\composition \emptyset\,\triangleright \psiCase{\varphi_e:\psiOutput{e}{e}.\psiAssertOp{\{e\}}} \stackrel{\overline{e}e}{\rightarrow}\psiAssertOp{\{e\}}
%    }
%    
%  }{
%    \emptyset\,\triangleright \psiCase{\varphi_e:\psiOutput{e}{e}.\psiAssertOp{\{e\}}}|Q\stackrel{\overline{e}e}{\rightarrow}\psiAssertOp{\{e\}}|Q
%  }
% \]

An event $e$ can happen if the corresponding condition in the \textbf{case} construct is entailed by the appropriate assertion $\Psi_Q\vdash\varphi_e$. This forms the right condition of the \textsc{(case)} rule, saying that all the preconditions of $e$ are met, and $e$ is not in conflict with any event that has happened. This condition is met because $C_\mathcal{E}=\Psi_Q$ and the assumption of the lemma, i.e., the existence of the step, which implies that $e$ is enabled by the configuration $C_\mathcal{E}$, meaning exactly what the definition of the entailment relation needs.

After \transition{\overline{e}e} has happened we have $P'= \psiAssertOp{\{e\}}|Q$ and $\psiframe{P'} = \psiframe{\psiAssertOp{\{e\}}}\composition \psiframe{Q} = \{e\}\cup \Psi_Q$, meaning that the frame of $P'$ corresponds to $C'_\mathcal{E}=C_\mathcal{E}\cup\{e\}$.
From the definition of the translation function $\evtopsi$ it is easy to see that $\evtopsi(\mathcal{E},C_\mathcal{E}')=\psiAssertOp{\{e\}}|Q$.

The second part of the lemma is especially easy after going through the proofs of the next results.
\end{proof}

\begin{theorem}[preserving concurrency]\label{th_preserv_conc}
For an event structure $\mathcal{E}\!=\!(E, \condRelEv, \conflictRelEv)$ with two concurrent events $e\concurrent e'$ then in the translation $\evtopsi(\mathcal{E},\emptyset)$ we find the behaviour forming the interleaving diamond, i.e., there exists $C_{\mathcal{E}}$ s.t.\ $\emptyset \semanticsPsi
\evtopsi(\mathcal{E},C_{\mathcal{E}})\transition{e} P_{1}\transition{e'}P_{2}$ and $\emptyset \semanticsPsi\evtopsi(\mathcal{E},C_{\mathcal{E}})\transition{e'}P_{3}\transition{e}P_{4}$ with $P_{2}=P_{4}$.
\end{theorem}

% \cp{This proof needs taking care of all details. In particular we need to consider what it means that two events are concurrent or causal in terms of configurations, i.e., definitions at the end of section 2.}
\vspace{-1ex}
\begin{proof}
In a prime event structure if two events $e,e'$ are concurrent then there exists a configuration $C$ reachable from the root which contains the conditions of both events, i.e., $\precondset{e}\subseteq C$ and $\precondset{e'}\subseteq C$, and does not contain any of the two events, i.e., $e,e'\not\in C$ (cf.\ Remark~\ref{remark_ES}).
Take this configuration as the one $C_{\mathcal{E}}$ sought in the theorem.
Therefore we have the following steps in the event structure: $C_{\mathcal{E}}\transition{e}C_{\mathcal{E}}\cup e$, $C_{\mathcal{E}}\transition{e'}C_{\mathcal{E}}\cup e'$, $C_{\mathcal{E}}\cup e\transition{e'}C_{\mathcal{E}}\cup\{e,e'\}$, and $C_{\mathcal{E}}\cup e'\transition{e}C_{\mathcal{E}}\cup\{e,e'\}$.

Since $C_{\mathcal{E}}$ is reachable from the root then by Lemma~\ref{lemma_conf_maint2} all the steps are preserved in the behaviour of the \eventpsi-process $\evtopsi(\mathcal{E},\emptyset)$, meaning that $\evtopsi(\mathcal{E},C_{\mathcal{E}})$ is reachable from  (i.e., part of the behaviour of) $\evtopsi(\mathcal{E},\emptyset)$.

% We also have that each sub process of the event-psi gives either empty set to its frame, or its event name, depending on whether it has happened or not. 

Since $e,e'\not\in C_{\mathcal{E}}$ we have that $\evtopsi(\mathcal{E},C_{\mathcal{E}})$ is in the form 
$P_0 = P_e|P_{e'}|Q$ with $P_{e}$ and $P_{e'}$ processes of kind \textbf{case}. From Lemma~\ref{conf_maint} we know that the frame of $\evtopsi(\mathcal{E},C_{\mathcal{E}})$ is the assertion corresponding to $C_{\mathcal{E}}$, which is $\psiframe{P_e|P_{e'}|Q} = \{\emptyset\}\cup\{\emptyset\}\cup\Psi_Q = \Psi_Q$.

%\hn{We have that $e,e'\notin \Psi_Q$ and $P_0\transition{e}P_1$ and $P_0\transition{e'}P_3$. From Lemma~\ref{lemma_conf_maint2} we have for these two transitions to be met the precondition for e and e' respectably must be met. And as $e,e'\notin\Psi_Q$ we have that $e'\notin \pi_L(\varphi_e)$ and $e\notin\pi_L(\varphi_{e'})$ as $\pi_L(\varphi_e)$ is the same as the set \precondset{e} and $\pi_L(\varphi_e')$ the set \precondset{e'} we have that the parts of definition of Casual Independence (concurrency) events in prime event structures $d\concurrent e \defequal \neg(d\condRelEv e\vee e\condRelEv d\vee d\conflictRelEv e)$ that concerns \condRelEv is met. 
%From Lemma~\ref{lemma_conf_maint2} we have that $P_1 = \psiAssertOp{e}|P_{e'}|Q$ and $P_3 = P_e|\psiAssertOp{e'}|Q$ We thus have that $e\in\psiframe{P_1}$ and $e'\in\psiframe{P_3}$. For the transition $P_1\transition{e'}P_2$ to happen we must have that $e\notin\pi_R(\varphi_{e'})$ and for $P_3\transition{e}P_4$ we must have $e'\notin\pi_R(\varphi_e)$. This is the same as $e'\notin\preconfset{e}$ and $e\notin\preconfset{e'}$. We have thus said that for these transitions to happen we must have that e and e' uphold the definition of Casual Independence in prime event structures. From the definition of \evtopsi we have that the conditions of both e and e' maintain they're relations when being mapped from ES to Event Psi, and thus we have that the conditions must uphold the demands for Casual Independence.}

From Lemma~\ref{lemma_conf_maint2} we see the transitions between the \eventpsi-processes:
$\emptyset \semanticsPsi \evtopsi(\mathcal{E},C_{\mathcal{E}})\transition{e} P_{1}\transition{e'}P_{2}$ with $P_{2}=\psiAssertOp{e}\psiPar\psiAssertOp{e'}\psiPar Q$ as well as $\emptyset \semanticsPsi\evtopsi(\mathcal{E},C_{\mathcal{E}})\transition{e'}P_{3}\transition{e}P_{4}$ with $P_{4}=\psiAssertOp{e}\psiPar\psiAssertOp{e'}\psiPar Q$.
We thus have the expected interleaving diamond. 

As a side, remark that $\psiframe{P_1} = \psiframe{P_0}\composition \psiAssertOp{e}$ and $\psiframe{P_3} = \psiframe{P_0}\composition\psiAssertOp{e'}$ thus $\psiframe{P_1}\composition\psiframe{P_3} = \psiframe{P_0}\composition\psiAssertOp{e}\composition\psiAssertOp{e'}=\psiframe{P_4}$, which say that $e\in \psiframe{P_1} \wedge e'\notin \psiframe{P_1}\wedge e'\in \psiframe{P_3}\wedge e\notin \psiframe{P_3}\wedge\psiframe{P_1}\composition\psiframe{P_3} = \psiframe{P_4}$. Using Lemma~\ref{conf_maint} these can be correlated with configurations and thus we can see the definition of concurrency from configurations as in Remark~\ref{remark_ES}.3.
\end{proof}

The proof of Theorem~\ref{th_preserv_conc} hints at an opposite result, stating a true concurrency rule for \eventpsi-processes. Intuitively the next result says that any two events that in the behaviour of the \eventpsi-process make up the interleaving diamond are concurrent in the corresponding event structure.
% This comes from the side-product of the proofs where we proof that only events that uphold the true concurrency rules for event structures can create an interleaving diamond in ES-Psi.  

\begin{theorem}[interleaving diamonds]\label{cor_independenceDiamonds}
For any event structure $\mathcal{E}$, in the corresponding \eventpsi-pro\-cess $\evtopsi(\mathcal{E},\emptyset)$, for any  interleaving diamond
$\emptyset \semanticsPsi
\evtopsi(\mathcal{E},C_{\mathcal{E}})\transition{e} P_{1}\transition{e'}P_{2}$ and $\emptyset \semanticsPsi\evtopsi(\mathcal{E},C_{\mathcal{E}})\transition{e'}P_{3}\transition{e}P_{4}$ with $P_{2}=P_{4}$, 
for some configuration $C_{\mathcal{E}}\in\mathcal{C}_{\mathcal{E}}$,
 we have that the events $e\concurrent e'$ are concurrent in $\mathcal{E}$.
\end{theorem}

\vspace{-1ex}
\begin{proof}
% From the proof of Theorem~\ref{th_preserv_conc} we have that for any concurrency diamond a similar diamond will be found in event psi. We also show that any interleaving diamond we get in event psi will make the frames of the psi processes fit exactly into the definition of concurrent events in event structures. And from lemma~\ref{lemma_conf_maint2} we have that the frames can be directly translated to the configurations in event structures definition of concurrent events. 
% 
Since $\evtopsi(\mathcal{E},C_{\mathcal{E}})$ has two outgoing transitions labelled with the events $e$ and $e'$ it means that $\evtopsi(\mathcal{E},C_{\mathcal{E}})$ is in the form 
$P_0 = P_e|P_{e'}|Q$ with $P_{e}$ and $P_{e'}$ processes of kind \textbf{case}. From Lemma~\ref{conf_maint} we know that the frame of $\evtopsi(\mathcal{E},C_{\mathcal{E}})$ is the assertion corresponding to $C_{\mathcal{E}}$, which is $\psiframe{P_e|P_{e'}|Q} = \{\emptyset\}\cup\{\emptyset\}\cup\Psi_Q = \Psi_Q$.

We thus have that $e,e'\notin \Psi_Q$ and $P_0\transition{e}P_1$ and $P_0\transition{e'}P_3$. This means that for these two transitions to be possible it must be that the precondition for $e$ and $e'$ respectably must be met. Since $e,e'\notin\Psi_Q$ it must be that $e'\notin \pi_L(\varphi_e)$ and $e\notin\pi_L(\varphi_{e'})$. Since $\pi_L(\varphi_e)$ is the same as the set \precondset{e} and $\pi_L(\varphi_e')$ the set \precondset{e'} we have the two parts of the Definition~\ref{def_conc_ES}  that concern \condRelEv\ for the casual independence (concurrency) of the events $e,e'$, i.e., $\neg(e'\condRelEv e\vee e'\condRelEv e)$. 
After the two transitions are taken we have that $P_1 = \psiAssertOp{e}|P_{e'}|Q$ and $P_3 = P_e|\psiAssertOp{e'}|Q$. We thus have that $e\in\psiframe{P_1}$ and $e'\in\psiframe{P_3}$. For the transition $P_1\transition{e'}P_2$ to happen we must have that $e\notin\pi_R(\varphi_{e'})$ and for $P_3\transition{e}P_4$ we must have $e'\notin\pi_R(\varphi_e)$. This is the same as $e'\notin\preconfset{e}$ and $e\notin\preconfset{e'}$ which makes the last part of Definition~\ref{def_conc_ES} concerning the conflict relation, i.e., $\neg(e'\conflictRelEv e)$.
This completes the proof, showing $e\concurrent e'$.
\end{proof}

We have seen that the \eventpsi-processes that we obtain from event structures in Definition~\ref{def_evtopsi} have a specific syntactic form. But the \eventpsi\ instance allows any process term to be constructed over the three nominal data-types that we gave in Definition~\ref{def_event_psi}. The question is which of all these \eventpsi-processes correspond exactly to event structures? We want to have syntactic restrictions on how to write \eventpsi-process terms so that we are sure that there exists an event structure corresponding to each such restricted process term.

\begin{theorem}[syntactic restrictions]\label{th_syntactic_restrictions}
Consider \eventpsi-process terms built only with the following grammar:
\vspace{-2ex}\[
P_{ES} :=  \psiAssertOp{e} \mid \psiCase{\varphi:\psiOutput{e}{e}.\psiAssertOp{e}} \mid P_{ES}\psiPar P_{ES}
\]
% \[
% P_{ES} :=  P_{e} \mid P_{case} \mid P_{ES}\psiPar P_{ES}
% \]
% \[
% P_{e} :=  \psiAssertOp{e}
% \]
% \[
% P_{case} :=  \psiCase{\varphi:\psiOutput{e}{e}.\psiAssertOp{e}}
% \]
Moreover, a term $P_{ES}$ has to respect the following constraints, for any $\varphi_{e},\varphi_{e'}$ from $\psiCase{\varphi_{e}:\psiOutput{e}{e}.\psiAssertOp{e}}$ respectively $\psiCase{\varphi_{e'}:\psiOutput{e'}{e'}.\psiAssertOp{e'}}$:
\begin{enumerate}
\item conflict: $e\not\in\pi_R(\varphi_{e})$ and $e'\in\pi_{R}(\varphi_{e})$ iff $e\in\pi_{R}(\varphi_{e'})$;

\item causality: $e\not\in\pi_L(\varphi_{e})$ and if $e\in\pi_{L}(\varphi_{e'})$ then $e'\not\in\pi_{L}(\varphi_{e})\wedge \pi_L(\varphi_{e})\subset \pi_L(\varphi_{e'})$;

\item executed events: $P_{ES}$ cannot have both $\psiAssertOp{e}$ and $\psiCase{\varphi:\psiOutput{e}{e}.\psiAssertOp{e}}$ for any $e$, nor multiples of each.
\end{enumerate}

\noindent For any such restricted process $P_{ES}$ there exists an event structure $\mathcal{E}$ and configuration $C_{\mathcal{E}}\in\mathcal{C}_{\mathcal{E}}$ s.t.
\vspace{-1ex}\[
\evtopsi(\mathcal{E},C_{\mathcal{E}})=P_{ES}.
\]
\end{theorem}

\vspace{-1ex}
\begin{proof}
From a \eventpsi-process $P_{ES}$ defined as in the statement of the theorem, we show how to construct an event structure $\mathcal{E}=(E,\condRelEv,\conflictRelEv)$ and a configuration $C_{\mathcal{E}}$. 
We have that $P_{ES}$ is built up of assertion processes and case guarded outputs, i.e., $P_{ES}=(\psiPar_{e\in E_{c}}\psiAssertOp{e}) \ \ \psiPar \  (\psiPar_{f\in E_{r}}\psiCase{\varphi_{f}:\psiOutput{f}{f}.\psiAssertOp{f}})$. 

Because of the third restriction on $P_{ES}$ we know that $E_{c}$ and $E_{r}$ are sets, as no multiples of the same process can exist. Moreover, these two sets are disjoint.
For otherwise, assume we have $\psiAssertOp{e} | \psiCase{\varphi_{e}:\psiOutput{e}{e}.\psiAssertOp{e}}$ part of $P_{ES}$. This is the same as if $e$ has happened already and $e$ may happen in future, which cannot be the case for event structures. 
% As no event may happen more than once in ES we cannot both have it happened and not happened same time, and assumption is wrong. Same we can do with assumptions of more than one $\psiAssertOp{e}$ in parallel composition as an event cant have happened more than once, and for $P_{case}$ for e as we cant have it possible to happen several times. And constrain 3. executed events must hold. 
% 
% We make the set $\mathcal{C} = E_c = \mathcal{F}P$ as per \ref{conf_maint}\ref{lemma_conf_maint2} we know that frame of a process and the configuration is the same. 

We take $C_{\mathcal{E}}$ to be the frame of $\psiframe{P_{ES}}=E_{c}$.
We take the set of events to be $E = E_c\cup E_r$.
We construct the causality and conflict relations from the processes in the second part of $P_{ES}$ as follows:
$\leq = \cup_{e\in E_r} \{(e',e) | e'\in \pi_L(\varphi_e)\}$ and 
$\sharp = \cup_{e\in E_r} \{(e',e)|e' \in \pi_R(\varphi_e)\}$. 
We prove that the causality relation is a partial order. 
For irreflexivity just use the first part of the second restriction on $P_{ES}$. For antisymmetry assume that $e\leq e' \wedge e'\leq e \wedge e\neq e'$ which is the same as having $e\in\pi_L(\varphi_{e'}) \wedge e'\in\pi_L(\varphi_{e})$. This contradicts the second restriction on $P_{ES}$. 
Transitivity is easy to obtain from the second restriction which says that when $e\leq e'$ then all the conditions of $e$ are a subset of the conditions of $e'$.
We prove that the conflict relation is irreflexive and symmetric. The irreflexivity follows from the first part of the first restriction on $P_{ES}$, whereas the symmetry is given by the second part.

It is easy to see that for the constructed event structure and the configuration chosen above, we have $\evtopsi(\mathcal{E},C_{\mathcal{E}})=P_{ES}$. The encoding function \evtopsi\ takes all events from $C_{\mathcal{E}}$ to the left part of the $P_{ES}$, whereas the remaining events, i.e., from $E_{r}$ are taken to \psiCase{}\!processes where for each event $f\in E_{r}$ the corresponding condition $\varphi_{f}$ contains the causing events respectively the conflicting events. But these correspond to how we built the two relations above.
\end{proof}

\subsection{Refinement}\label{subsec_refinement}

We want to be able to refine psi processes on the same line as labelled event structures are refined in \cite{GlabbeekG01refinement}. We recall below the definition of refinement of event structures from \cite{GlabbeekG01refinement}.

A \textit{refinement function} \refinement, is a function from actions to event structures without conflict (i.e., the conflict relation is empty). This is considered as a given function to be used in the \textit{refinement operation}. This refinement operation can be also seen as a function from event structures together with functions as above, and returning new event structures, i.e., like an algorithm.
For notation economy this algorithm is also denoted by $ref$, to connect it with the essential input it takes as the refinement function $\refinement : Act \rightarrow \textbf{E}_{\not\,\,\conflictRelEv}$ (with $\textbf{E}_{\not\,\,\conflictRelEv}$ denoting conflict-free prime event structures).

\begin{definition}[refinement for prime event structures]\label{def_ref_ES} For an event structure $\mathcal{E}$ with events labelled by $l: E \rightarrow Act$ with actions from $Act$ we have the following definitions.

(i) A Function $\refinement : Act \rightarrow \textbf{E}_{\not\,\,\conflictRelEv}$ is called a \emph{refinement function} (for prime event structures) iff $\forall a\in Act : \refinement(a)$ is a non-empty, finite and conflict-free labelled prime event structure.

(ii) Let $\mathcal{E} \in \textbf{E}_{}$ and let $\refinement$ be a refinement function.

Then $\refinement(\mathcal{E})$ is the prime event structure defined by:
\begin{itemize}
\item $E_{\refinement(\mathcal{E})}:= \{(e,e')|e\in E_\mathcal{E},e' \in E_{\refinement(l_\mathcal{E}(e))}\},$ where $E_{\refinement(l_\mathcal{E}(e))}$ denotes the set of events of the event structure $\refinement(l_\mathcal{E}(e))$,

\item $(d,d')\condRelEv_{\refinement(\mathcal{E})}(e,e')$ iff $d \condRelEv_\mathcal{E} e$ or $ (d= e\wedge d' \condRelEv_{\refinement(l_\mathcal{E}(d))}e')$,

\item $(d,d')\sharp_{\refinement(\mathcal{E})} (e,e')$ iff $ d\sharp_\mathcal{E} e$,

\item $l_{\refinement(\mathcal{E})}(e,e'):=l_{\refinement(l_\mathcal{E}(e))}(e')$. 
\end{itemize}
\end{definition}

The intuition of refinement is to take one action (which is thought as an abstraction) and give it more structure. Since the same action can be instantiated several times at different points in the system, i.e., by different events, all these events labelled by the same action are given more structure by replacing them with a new event structure. For example one event can become a sequence of events, or the parallel composition of deterministic components.
But refinement is restricted to not contain conflicts, i.e., not contain choices. This is because of technical reasons that make it not possible to define the new conflict relation so to obtain prime event structures after refinement. But there are also natural counter-examples for requiring conflict-free refining event structures, and van Glabbeek and Goltz in \cite{GlabbeekG01refinement} explain these much better than we ever could.
We need a similar refinement operation for \eventpsi-process terms.

\begin{definition}
Given a refinement function for event structures \refinement, we define an operation $\refinement^\Psi$ that refines an \eventpsi-process to a new one over the names
\[
T^\Psi = \{(e,e') \mid e\in E, e'\in E_{ref(l(e))}\}. 
\]
An \eventpsi-process $P$, build according to Theorem~\ref{th_syntactic_restrictions}, with frame $\psiframe{P}=\Psi_P$, is refined into a process 
\[
\hspace{9ex}\refinement^{\psi}\!(P)=|_{(e,e')\in T^P}P_{(e,e')}, \text{\hspace{8ex} with } T^P = \{(e,e')|e\in \mathit{en}(P), e'\in  E_{ref(l(e))}\}
\]
and 
$P_{(e,e')} = \psiAssertOp{\{(e,e')\}} $, if $e \in \Psi_P $, otherwise
$P_{(e,e')} = \psiCase{\varphi_{(e,e')}:\overline{(e,e')}(e,e').\psiAssertOp{\{(e,e')\}}}$, with the conditions being 
\[
\varphi_{(e,e')} = (\precondset{(e,e')}, \preconfset{(e,e')}),
\]
where\ \ $\precondset{(e,e')} = \{(d,d') \mid d\in\pi_L(\varphi_e) \vee (d=e \wedge d'\in\condRelEv_{ref(l(d))} e)\}$\ \ and\ \ $\preconfset{(e,e')} = \{(d,d')| d \in \pi_R(\varphi_e)\}$.
% 
% \hn{In the end we rename any $(d,d')\in T^P$ to an $e\in E$ respecting the order.}
\end{definition}

% \cp{
% The old definition.
% 
% 
% \begin{definition}
% We define a function $\refinement^P$ that refines an \eventpsi-process to a new one over
% \[
% T^P = \{(e,e') \mid e\in T, e'\in T_{ref(e)}\},
% \]
% with the conditions being 
% \[
% C^P = \{(\precondset{(e,e')}, \preconfset{(e,e')}) \mid (e,e') \in T^P\},
% \]
% where $\precondset{(e,e')} = \{(d,d') \mid d\in\precondset{e} \vee (d=e \wedge d'\in\condRelEv_{ref(d)} e)\}$ and 
% $\preconfset{(e,e')} = \{(d,d')| d \in \preconfset{e}\}$ to obtain 
% \[
% P_{ref} = |_{(e,e')\in T^P}P_{(e,e')}
% \]
% with 
% $P_{(e,e')} = \psiAssertOp{\{(e,e')\}} $ if $e \in \Psi_P $, otherwise
% $P_{(e,e')} = \psiCase{\varphi_{(e,e')}:\overline{(e,e')}(e,e').\psiAssertOp{\{(e,e')\}}}$.
% \end{definition}
% }

The new names are pairs of a parent event name (i.e., from the original process) and one of the event names from the refinement processes. 
% % This can be the same as the parents name. 
% We have that each parents name is unique, and each children's name is also unique we will get new unique names using the two names to make a new one consisting of both. These new names can be looked at as single elements when the new psi-process is made and not as pairs. 
% We do not end up outside the event psi instance because $T^{P}$ is isomorphic to $T$. 
We do not end up outside the \eventpsi\ instance because we can rename any pair by names from $E$. Take any total order $<$ on $E$ and define from it a total order $(e,e') < (d,d') \mbox{ iff } e<d\vee (e=d\wedge e'<d')$ on the pairs; rename any pair by an event from $E$ while preserving the order, thus making $T^{\psi}$ the same as the \psiTerms\ of \eventpsi.
\cp{In fact it is NOT isomorphic! So maybe we need to change the definition of event-psi instance. Think about defining the names to be the natural numbers, or any set that can be identified to the natural numbers, or with a finite subset of them. Then I could say that the pairs of names can still be identified with the natural numbers.}

We make new conditions for each of the new names $(e,e')$, where \precondset{(e,e')} contains all pairs of names s.t.\ either the left part is a condition for $e$, or the left part is the same as $e$ but the right part is a condition for $e'$. The conflicts set \preconfset{(e,e')} contains all pairs of names with the first part a conflict for $e$. The refinement generates for each new pair one process which is either an assertion or a \psiCase{}process, depending on whether the first part of the event pair was in the frame of the old $P$ or not.

%We makes all names to be a pair of names. Either the parent events name with one of its children events name, or as a pair of two parent names. And do this for all possible pairs, just as we make new events when we refine an event structure. 

%We also make new conditions for the new names, where the set of the names conditions for e are either all pairs where first part of the name was a condition for first part of e's name, or first parts is equal and second part is a condition for second part of e. This corresponds to the rule for expanding the conditions in ref of prime event structures. For conflicts we have that set of conflicts for e is the set of names where first part of name was a conflict to first part of e's name. also in conjunction with definition of ref. 

%We make then new processes for each of the new names, either it will be in non executed state if first part of e's name is not in frame of old event-psi, and is in executed state if it was in frame of old event-psi. This is to make it fit in with fact that we can have each process in 2 distinct states, giving the frame of an event, who then corresponds to the configuration of the ES. As ref is defined i assume that the configuration and thus the frame of its event-psi is $\emptyset$.

\begin{theorem}[refinement in \eventpsi\ corresponds to refinement in ES]\label{refineOfPsi}
For any prime event struc\-ture $\mathcal{E}$ we have that: \hspace{3ex} $\evtopsi(\mathit{ref}(\mathcal{E}),\emptyset) = \mathit{ref}^{\psi}(\evtopsi(\mathcal{E},\emptyset))$.
\end{theorem}

\begin{proof}
As $\psiTerms = E$ and as $T^{\psi}$ is built from \psiTerms\ with the same rules as $E_{\refinement}$ is built from $E$ we have that $T^{\psi} = E_{\refinement}$.
Since the processes we work with are parallel compositions of assertion and \psiCase{}processes, it means we have to show that any assertion processes on the left is also found on the right of the equality (and vice versa), and the same for the \psiCase{}processes. Since we work with the empty initial configuration, then there are no assertion processes on neither sides.

The \psiCase{}\!processes on the left side are those generated by \evtopsi\ from the pairs events returned by the $\mathit{ref}$ from the event structure. This means that for each pair we have its condition built up as in the Definition~\ref{def_ref_ES}.
On the right side we have \psiCase{}processes for the original process before the refinement, with their respective conditions. But the $\mathit{ref}^{\Psi}$ replaces these with many \psiCase{}processes, one for each new pair, and for each the conditions are build exactly as the $\mathit{ref}$ is defining them. This says that we have the same number of \psiCase{}processes on both sides of the equality, and they have the same conditions.
% 
% The conditions and conflicts, as with the \psiTerms\ is build up with the same rules for who is a condition to who, and who is in conflict, in both $\refinement$ and $\refinement^P$, as $\pi_R(\varphi_e) = \precondset{e}$ and $\pi_L(\varphi_e) = \preconfset{e}$. From the $<_\refinement$ we get when turning to \eventpsi\ for all $e\in E$ the set of events they has as condition and the set they has as conflicts. This are the same sets we make with $\refinement^P$
% 
% For $\refinement$ we assume that it has no configuration (its empty) as no definition of changing the configuration is given. So that means that the frame of \eventpsi\ will be empty. Now for all events in $E_{\refinement}$ will be given the form of an non executed event with its $\varphi$. For $\refinement^P$ we do the same for all events in $T^P$ as $T^P = E_{\refinement}$ we have that we make the same processes. 
\end{proof}

\section{DCR graphs as psi-calculi}\label{sec_DCR_into_psi}

We achieved a rather natural and intuitive translation of the prime event structures into an instance of psi-calculi. We made special use of the logic of psi-calculi, i.e., of the assertions and conditions and the entailment between these, as well as the assertion processes. Noteworthy is that we have not used the communication mechanism of psi-calculus, which is known to increase expressiveness.

%We try to extend this approach from event structures to the DCRs. But it appears that we need the communication constructs on processes to keep track of the current marking of a DCR. The particularities and expressiveness of DCRs do not allow for a simple way of updating the marking, as was the case for event structures when just union with the newly executed event was enough. But once we use the communication, outputing a term representing the current marking, we manage to get a natural encoding of the DCRs in a psi-calculus instance for these. We can then see correlations with the previous encoding of the event structures. The markings are kept in the assertions, i.e., as the frame of the process; the same as we did with the configurations of the event structures. Case processes are used for each event of the DCR, and the conditions of the case processes capture the conditions that the events of a DCR depend on to be enabled in a marking. The entailment relation then captures the enabling of events by markings. Since union is not enough to capture update of the marking, we incorporate an idea of generation (or age) in an assertion, and assertion composition keeps the oldest generation. In this way the requirements for the assertion composition can still be met, so to yield a psi-calculus instance.
We try to extend this approach from event structures to the DCRs. But it appears that we need the communication constructs on processes to keep track of the current marking of a DCR. The particularities and expressiveness of DCRs do not allow for a simple way of updating the marking, as was the case for event structures when just union with the newly executed event was enough. But once we use the communication, outputting a term representing the current marking, and incorporating an idea of generation (or age) of an assertion, where assertion composition keeps the newest generation which would be used for entailments, we get a nice natural encoding for DCRs in a psi-calculus instance.  We can then see associations with the previous encoding of the event structures. The markings are kept in the assertions, i.e., as the frame of the process; the same as we did with the configurations of the event structures. Case processes are used for each event of the DCR, and the conditions of the case processes capture the information needed to decide when events of a DCR are enabled in a marking. The entailment relation then captures the enabling of events.

\begin{definition}[\dcrpsi\ instance]\label{def_DCRpsi_istance}
We define an instantiation of Psi-calculi called \dcrpsi\ by providing the following definitions:
\vspace{-2ex}$$\psiTerms \eqbydef \{m\} \cup \psiAssertions \vspace{-1ex}$$
$$\psiAssertions \eqbydef \partsof{E}\times \partsof{E}\times \partsof{E}\times \mathbb{N}\vspace{-1ex}$$
where $E$ is a nominal set and $\mathbb{N}$ is the nominal data structure capturing natural numbers using a successor function $s(\cdot)$ and generator $0$, whereas $m$ is a single name used for communication;
% 
% 
% $$C= P(E)\times P(E)\times E\times P(E)$$
$$\psiConditions \eqbydef \partsof{E}\times \partsof{E}\times E \hspace{6ex} 
\psiChanEq \eqbydef = \hspace{6ex} 
\psiAssertUnit \eqbydef (\emptyset,\emptyset,\emptyset,0)\vspace{-1ex}$$
% 
% \vspace{-1ex}
% $$\psiChanEq \eqbydef =$$
% 
% \vspace{-1ex}
% $$\psiAssertUnit \eqbydef (\emptyset,\emptyset,\emptyset,0)$$
% 
% \vspace{-3ex}
% $$ \composition \eqbydef \Psi\composition\Psi' \text{ written as } \psiAssertOp{(Ex, Re, In, G)}\composition \psiAssertOp{(Ex', Re', In', G')}$$ 
$$\psiAssertOp{(Ex, Re, In, G)}\composition \psiAssertOp{(Ex', Re', In', G')}\eqbydef
\begin{cases} \psiAssertOp{(Ex, Re, In, G)} & \text{if $G>G'$} \\
        \psiAssertOp{(Ex', Re', In', G')} & \text{if $G<G'$} \\
        \psiAssertOp{(Ex\cup Ex', Re\cup Re', In\cup In', G)} & \text{if $G=G'$}\end{cases}$$
where the comparison $G<G'$ is done using subterm relation, eg., $s(N)>N$.
Entailment \psiEntailment\ is defined as: 
% \vspace{-1ex}
$$\psiAssertOp{(Ex, Re, In, G)} \vdash (Co, Mi, e) \mbox{\ \ \ iff\ \ \ } e\in In \wedge (In\cap Co)\subseteq Ex \wedge ((In\cap Mi)\cap Re)=\emptyset.$$
\end{definition}

% We have that the Terms can be Events in Dcr and the Natural numbers. 
Terms can be either a name $m$, which we will use for communications, or assertions which will be the data communicated.
Assertions are a tuple of three sets of events, and a number we intend to hold the \textit{generation} of the assertion. 
The first set is meant to capture what events have been executed, the second set for those events that are pending responses, and the third set for those events that are included. These three sets mimic the same sets that the marking of a DCR-graph contains.
% The generation is here so when we have one assertion is composed with another its only the newest assertion that is kept with all its sets. This is a way that allows composition to uphold the rules of psi-calculi while same letting us demand only the last assertion to be kept. 
The generation number is used to get the properties of the assertion composition, which are somewhat symmetric, but still have the composition return only the latest marking/assertion (i.e., somewhat asymmetric).

The composition of two assertions keeps the assertion with highest generation.\footnote{For technical reasons, when we compose two assertions with the same generation number we obtain an assertion where the sets are the union between the associated sets in each assertion, and the generation number is unchanged.} This makes the composition associative, commutative, compositional, and with identity defined to be the tuple with empty sets and lowest possible generation number.

The conditions are tuples of two sets of events and a single event as the third tuple component. The first set is intended to capture the set of events that are conditions for the single event. The second set is intended to capture the set of events that are milestones for the single event. 
% The single event is the event who's trailing receiving when reduced will mark the execution of the event in the dcr-graph. 

% We have that composition of two assertion looks at the generation number of the assertion and leaves behind the data from the assertion with highest generation along with this generation number, if they have the same generation we leave behind and assertion where the sets are the union between the two associated sets in each assertion, and leave the generation behind. This makes it be associative, commutative, compositional and as identity is defined to be a tuple with empty sets and lowest possible generation we have that it maintains identity.

% Channel equivalence is equality on names.

% An assertion entails a condition if the event that may happen if this entailment hold is included in the set of executed events in the assertion. The set of included condition events is in the set of executed events. The set of included milestones are in the set of non pending events. These are the same rules that are saying if an event is enabled in a Dcr-graph. 
The entailment definition mimics the definition in DCR graphs for when an event (i.e., the third component of the conditions) is enabled in a marking (i.e., the first three components of the assertions). Compare the example below with the definition of enabling from DCR graphs\vspace{-1ex}
$$\psiAssertOp{(Ex, Re, In, G)} \vdash (\conditionRel e, \milestoneRel e, e) \mbox{\ \ \ iff\ \ \ } e\in In \wedge (In\cap \conditionRel e)\subseteq Ex \wedge ((In\cap\milestoneRel e)\cap Re)=\emptyset.$$

\begin{definition}
We define the function \dcrtopsi\ which takes a DCR $(E, M \conditionRel, \responseRel, \milestoneRel, \includeRel, \excludeRel, L, l)$ with distinguished marking $M=(Ex', Re', In')$ and returns a \dcrpsi\ process\vspace{-1ex} 
$$P_{dcr} = P_s \psiPar P_E \vspace{-2ex}$$
where
\vspace{-2ex}
$$P_s = \psiAssertOp{(Ex',Re',In',0)} \psiPar \psiOutput{m}{(Ex',Re',In',0)}.\psiEmptyProc \hspace{4ex}\mbox{and}\hspace{4ex}
P_E = \psiPar_{e\in E} P_e$$ 
\vspace{-1ex}
with
\vspace{-1ex}
$$P_e = !(\psiCase{\varphi_e:\psiInput{m}{(X_E,X_R,X_I,X_G)}.$$
$$(\psiOutput{m}{(X_E\cup\{e\},(X_R\setminus\{e\})\cup e\responseRel,(X_I\setminus e\excludeRel)\cup e\includeRel, s(X_G))}.\psiEmptyProc \psiPar $$
$$\psiAssertOp{(X_E\cup\{e\},(X_R\setminus\{e\})\cup e\responseRel,(X_I\setminus e\excludeRel )\cup e\includeRel, s(X_G))})})$$
where $X_E,X_R,X_I,X_G$ are variables and \hspace{2ex} $\varphi_e = (\conditionRel e, \milestoneRel e, e)$.
\end{definition}

The process $P_{dcr}$ generated by \dcrtopsi\ contains a starting processes $P_{s}$ that models the initial marking of the encoded DCR as an assertion process, and also communicates this assertion on the channel $m$. The rest of the process, i.e., $P_{E}$ captures the actual DCR, being a parallel composition of processes $P_{e}$ for each of the events of the encoded DCR.
The events in a DCR can happen multiple times, hence the use of the replication operation as the outermost operator. Each event is encoded, following the ideas for event structures, using the \psiCase{}\!construct with a single guard $\varphi_e$. The guard contains the information for the event $e$ that need to be checked against the current marking (i.e., the assertion) to decide if the event is enabled; these information are the set of events that are prerequisites for $e$ (i.e., $\conditionRel e$) and the set of milestones related to $e$.
There may be several events enabled by a marking, hence several of the parallel \psiCase{}\!processes may have their guards entailed by the current assertion. Only one of these input actions will communicate with the single output action on $m$, and will receive in the four variables the current marking. After the communication, the input process will leave behind an assertion process containing an updated marking, and also a process ready to output on $m$ this updated marking. In fact, after a communication, what is left behind is something looking like a $P_{s}$ process, but with an updated marking. The updating of the marking follows the same definition from the DCRs.

\begin{lemma}\label{init-frame}
For any DCR graph $\mathcal{D}$, the frame of the corresponding process $\dcrtopsi(\mathcal{D})$ corresponds to the marking of the encoded DCR (i.e., the first three components). 
\end{lemma}

\vspace{-1ex}
\begin{proof}
$\dcrtopsi(\mathcal{D})$ return a \dcrpsi\ process with only one assertion which thus is the frame. This assertion is made directly from the marking of $\mathcal{D}$ and added generation 0. 
\end{proof}
\vspace{-1ex}

\begin{lemma}\label{onecomm}
For any DCR graph $\mathcal{D}$, in the execution graph of the corresponding process $\dcrtopsi(\mathcal{D})$ at any execution point there will be only one output process.
\end{lemma}

\vspace{-1ex}
\begin{proof}
Initially we have only one output in the $P_s$ part of $\dcrtopsi(\mathcal{D})$. Inductively we assume a reachable process $P$ with only one output process. If we have any enabled input processes only one of these processes will join a communication with the single output process. All input processes are of the form $P_e$, 
which reduces with psi rules for replication and input to\vspace{-1ex}
$$P_e|(\psiOutput{m}{(X_E\cup\{e\},(X_R\setminus\{e\})\cup\responseRel e,(X_I\setminus\excludeRel e)\cup\includeRel e, s(X_G))}.\psiEmptyProc \psiPar $$
$$\psiAssertOp{(X_E\cup\{e\},(X_R\setminus\{e\})\cup\responseRel e,(X_I\setminus e\excludeRel )\cup \includeRel e, s(X_G))})\vspace{-1ex}$$
with  $X_E, X_R, X_I, X_G$ substituted with the terms that were sent. The output process reduces to \psiEmptyProc. We have added as many new output processes as we have removed, and as we initially only have one output process by induction we always will have only one. 
\end{proof}

\begin{lemma}\label{frame-message_equal}
For any DCR graph $\mathcal{D}$, in the corresponding process $\dcrtopsi(\mathcal{D})$
the message being sent will always be the same as the frame of the \dcrpsi\ process. 
\end{lemma}

\vspace{-1ex}
\begin{proof}
Initially, the first message being sent by $P_{s}$ is by construction the same as the initial frame. 
% 
% On communication we have that a receiver is reduced to a new assertion and a new sender of proof of lemma \ref{onecomm}. The rules for this new assertion and the message of this new sender, are identical. As we are using the same variables for both the assertion and message also in same places. We have that they will always be same. 
The proof of Lemma~\ref{onecomm} shows that with each communication a new assertion is added and a new sender replaces the old one. The two new terms (i.e., the assertion process and the message) are identical and have the generation part increased by one. 
Since the composition of assertions keeps only the assertion with the higher generation, all older assertion processes that are still present are being ignored when computing the frame of the new process. We thus have our result.
\end{proof}

\vspace{-1ex}
\begin{lemma}[generations count transitions]\label{generation_count}
The generation part of the frame is the same as the number of transitions we have done from the initial process. 
\end{lemma}

\vspace{-1ex}
\begin{proof}
We use induction and assume we have done $n$ transitions and the generation part of our frame is $n'$ where $n=n'$. 
% (The basis of the induction is trivially ensure through the definition of $\dcrtopsi$.) 
From Lemma~\ref{frame-message_equal} we have that the frame and message are equal, so we will be sending $n$ as generation part of the message. After the communication a new assertion with generation $s(n')$ is added, which by the definition of assertion composition will be the new frame. By our assumption $s(n') = s(n) = n+1$. From Lemma~\ref{init-frame} we have that $n=n'=0$ for the initial process, and by induction we have that this holds for any number of transitions.
\end{proof}

\begin{theorem}[preserving transitions]\label{th_preserveDCRsteps}
In a DCR graph $\mathcal{D}$, for any transition $(\mathcal{D},M)\transition{e}(\mathcal{D},M')$ there exists a reduction between the corresponding \dcrpsi\ processes $\dcrtopsi(\mathcal{D},M)\transition{\tau}\dcrtopsi(\mathcal{D},M')$.
\end{theorem}

% \vspace{-1ex}
\begin{proof}
From Lemma \ref{init-frame} we know that the frame and marking are the same. This means that since $M\psiEntailment e$, the corresponding condition in the $\dcrtopsi(\mathcal{D},M)$ will be entailed by the frame. Therefore a communication is possible, i.e., a transition labelled by $\tau$. For $M = (Ex,Re,In)$ it means that the frame of $\dcrtopsi(\mathcal{D},M)$ is $(Ex, Re, In, G)$.
From Lemma \ref{frame-message_equal} we know that the frame is always the same as the message being sent. 
When the transition corresponding to the event $e$ happens the new frame of the \dcrpsi\ becomes
% \vspace{-1ex}
\[
\psiAssertOp{(Ex\cup\{e\},(Re\setminus\{e\})\cup\responseRel e,(In\setminus e\excludeRel )\cup e\includeRel, s(G))}
\] 
after alpha-conversion.
For a transition in DCR over the event $e$ we get the new marking 
% \vspace{-1ex}
\[M' = (Ex\cup\{e\}, (Re\setminus\{e\})\cup e\responseRel, (In\setminus e\excludeRel)\cup e\includeRel),
\]
which is the same as the new frame, with the exception of the generation part.
\end{proof}

Interesting would be to look closer at the encoding of event structures through the \evtopsi\ and the encoding through \dcrtopsi\ when seen as a special case of DCRs; a question on these lines would be: are $\evtopsi(ES)$ and $\dcrtopsi(DCR(ES))$ bisimilar?
First of all, \evtopsi\ translates into the \eventpsi\ instance, whereas \dcrtopsi\ into the \dcrpsi\ instance, and these two instances work with different terms and operator definitions. Even more, the encoding of event structures exhibits behaviour through \textit{labelled transitions}, whereas the behaviour of \dcrpsi\ encodings exposes \textit{only $\tau$-transitions}. Therefore, it is not easy to find a bisimulation-like correspondence.

Nevertheless, there are clear correlations. Consider an un-labelled event structure $(E,\condRelEv,\conflictRelEv)$ and its presentation as a DCR graph $(E,M,\condRelEv,\emptyset,\emptyset,\emptyset,\conflictRelEv\cup id)$ with the marking $M=(\emptyset,\emptyset,E)$; and denote the associated psi-processes by $P_{ES}=\evtopsi(ES)$ and $P_{DCR}=\dcrtopsi(DCR(ES))$. Correlate an assertion in $P_{DCR}$ with the assertion in $P_{ES}$ by looking only at the first set of the quadruple (having the second set of the quadruple, which encodes responses, always empty). The conditions of $P_{DCR}$ have the second set of milestones always empty; whereas the first set is the same as the first set of the conditions in $P_{ES}$.
One can now check that the entailment of a condition by an assertion in $P_{ES}$ is the same as the corresponding entailment in the $P_{DCR}$, when considering also the other behaviour aspects of these two processes and how they change the assertions.
But we do define this investigation to a longer version of this paper.

% \vspace{-3ex}
\section{Conclusions and outlook}\label{sec_conclusion}
% \vspace{-1ex}

We have encoded the true concurrency models of prime event structures and DCR graphs into corresponding instances of psi-calculi. For this we have made use of the expressive logic that psi-calculus provides to capture the causality and conflict relations of the prime event structures, as well as the relations of DCR-graphs. The computation in the concurrency models corresponds to reduction steps in the psi-processes. The more expressive model of DCR-graphs required us to make use of the communication mechanism of psi-calculi, whereas for event structures this was not needed. The data terms we sent were tuples of terms, capturing markings of DCR-graphs with a generation number attached to them.

For the encodings we also investigated some results meant to provide more confidence in their correctness. In particular, for event structures we also looked at action refinement as well as gave the syntactic restrictions that capture the psi-processes that exactly correspond to event structures.
Besides providing correlations between the computations in the respective models, we also investigated how true concurrency is correlated to the interleaving diamonds in the encodings we gave.

The purpose of our investigations was to see how well the expressiveness of psi-calculi can accommodate the expressiveness of true concurrency models.
Nevertheless, a discrepancy remains between the interleaving semantics based on SOS rules of psi-calculi, and the true concurrency nature of the two models we considered. Further investigations would look for a true concurrency semantics for psi-calculi (with initial results presented as \cite{NPH14MeMo}), and then see how our encodings fit with the true concurrency models that this semantics would return. One could also look into adding responses to psi-calculus, similar to how is done in \cite{DBLP:journals/corr/abs-1207-4270} for Transition Systems with Responses.

% \nocite{*}
\bibliographystyle{eptcs}
\bibliography{bib}

\end{document}